\documentclass[a4paper,UKenglish]{lipics-v2016}
 
\usepackage{microtype}
\usepackage{enumitem}
\newtheorem{observation}[theorem]{Observation}

\newcommand{\ceil}[1]{\lceil #1 \rceil}
\newcommand{\floor}[1]{\lfloor #1 \rfloor}

\newcommand{\dist}{\operatorname{dist}}

\newcommand{\CP}{\textsf{Closest Pair}}
\newcommand{\CPD}{\textsf{Closest Pair Decision}}

\newcommand{\linfty}{\mathrm L_\infty}

\def\reals{{\mathbb R}}

\newcommand{\Matousek}{Matou\v{s}ek}

\title{Dominance Product and High-Dimensional Closest Pair under $L_\infty$\footnote{Work on this paper has been supported 
		by Grant 892/13 from the Israel Science Foundation, 
		by Grant 2012/229 from the U.S.-Israeli Binational Science Foundation,
		by the Israeli Centers of Research Excellence (I-CORE)
		program (Center No.~4/11),
		and by the Hermann Minkowski--MINERVA Center for Geometry at Tel Aviv
		University.}}
\titlerunning{Dominance Product and High-Dimensional Closest Pair under $L_\infty$} 

\author[1]{Omer Gold}
\author[2]{Micha Sharir}
\affil[1]{Blavatnik School of Computer Science, Tel Aviv University,\\
	Tel Aviv 69978, Israel\\
  \texttt{omergold@post.tau.ac.il}}
\affil[2]{Blavatnik School of Computer Science, Tel Aviv University,\\
	Tel Aviv 69978, Israel\\
  \texttt{michas@post.tau.ac.il}}
\authorrunning{O. Gold and M. Sharir} 

\Copyright{Omer Gold and Micha Sharir}

\subjclass{F.2.2 Nonnumerical Algorithms and Problems}
\keywords{Closest Pair, Dominance Product, $L_\infty$ Proximity}

\ArticleNo{1}

\begin{document}

\maketitle

\begin{abstract}
Given a set $S$ of $n$ points in $\mathbb{R}^d$, the \CP{} problem is to find a pair of distinct points in $S$ at minimum distance.   
When $d$ is constant, there are efficient algorithms that solve this problem, and fast approximate solutions for general $d$.
However, obtaining an {\em exact} solution in very high dimensions seems to be much less understood.
We consider the high-dimensional $\linfty$ \CP{} problem, where $d=n^r$ for some $r > 0$, and the underlying metric is $L_\infty$.     

We improve and simplify previous results for $\linfty$ \CP{}, showing that it can be solved by a deterministic strongly-polynomial algorithm 
that runs in $O(DP(n,d)\log n)$ time, and by a randomized algorithm that runs in $O(DP(n,d))$ expected time,
where $DP(n,d)$ is the time bound for computing the {\em dominance product} for $n$ points in $\reals^d$.
That is a matrix $D$, such that
$D[i,j] = \bigl| \{k \mid p_i[k] \leq p_j[k]\} \bigr|$; this is the number of coordinates at which $p_j$ dominates $p_i$.
For {\em integer} coordinates from some interval $[-M, M]$, we obtain an algorithm that runs in 
$\tilde{O}\left(\min\{Mn^{\omega(1,r,1)},\, DP(n,d)\}\right)$ time\footnote{The $\tilde{O}(\cdot)$ notation hides poly-logarithmic factors.},
where $\omega(1,r,1)$ is the exponent of multiplying an $n \times n^r$ matrix by an $n^r \times n$ matrix.

We also give slightly better bounds for $DP(n,d)$, by using more recent rectangular matrix multiplication bounds. 
Computing the dominance product itself is an important task, since it is applied in many algorithms as a major black-box ingredient, such as algorithms for APBP (all pairs bottleneck paths),
and variants of APSP (all pairs shortest paths).
 \end{abstract}

\section{Introduction}\label{sec:intro}
Finding the closest pair among a set of $n$ points in $\reals^d$ was among the first studied algorithmic geometric problems, considered at the origins of computational geometry; see~\cite{Shamos75, Preparata85}.
The distance between pairs of points is often measured by the $L_\tau$ metric, for some $1\leq \tau \leq \infty$, under which the distance between the points
$p_i= (p_i[1], \ldots, p_i[d])$ and $p_j=(p_j[1],\ldots, p_j[d])$ is
$
\dist_\tau(p_i, p_j) = \| p_i - p_j \|_\tau = \left(\sum^{d}_{k=1}{\bigl| p_i[k] - p_j[k] \bigr|^\tau}\right)^{1/ \tau},
$
for $\tau < \infty$, and
$
\dist_\infty(p_i, p_j) = \| p_i - p_j \|_\infty = \max_{k} {\bigl| p_i[k] - p_j[k]\bigr|}, 
$ for $\tau =\infty$. 
The \CP{} problem and its corresponding {\em decision} variant, under the $L_\tau$-metric, are defined as follows.
\setdescription{leftmargin=1.5cm,labelindent=\parindent}
\begin{description}
	\item[{\bf \CP:}]
	Given a set $S$ of $n$ points in $\reals^d$, find a pair of distinct points $p_i ,p_j \in S$ such that
	$\dist_\tau(p_i, p_j) = \min_{\ell \neq m}\{ \dist_\tau(p_\ell, p_m) \mid  p_\ell, p_m \in S \}$.
	\item[{\bf \CPD:}]
	Given a set $S$ of $n$ points in $\reals^d$, and a parameter $\delta > 0$, determine whether there is a pair of distinct points $p_i ,p_j \in S$
	such that $\dist_\tau(p_i, p_j) \leq \delta $.
\end{description}
Throughout the paper, the notation $\mathrm{L_\tau}$ \CP{} refers to the \CP{} problem under some {\em specific} metric $L_\tau$, for $1 \leq \tau \leq \infty$
(and we will mostly consider the case $\tau= \infty$).

In the {\em algebraic computation tree model} (see~\cite{Ben-Or83}), the \CP{} problem has a complexity lower bound of $\Omega(n\log n)$ (for any $L_\tau$ metric),
even for the one-dimensional case $d=1$,
as implied from a lower bound for the Element-Uniqueness problem~\cite{Ben-Or83}.

As for upper bounds, Bentley and Shamos~\cite{Bentley76, Bentley80}
were the first who gave a deterministic algorithm for finding the closest pair under the $L_2$ metric that runs in $O(n\log n)$ time for any {\em constant} dimension $d \geq 1$, 
which is optimal in the algebraic computation tree model, for any fixed $d$.
Their algorithm uses the {\em divide-and-conquer} paradigm, and became since, a classical textbook example for this technique.
In 1976 Rabin presented, in a seminal paper~\cite{Rabin76}, a {\em randomized} algorithm that finds the closest pair in $O(n)$ expected time,
using the {\em floor} function (which is not included in the algebraic computation tree model).
His algorithm uses random sampling to decompose the problem into smaller subproblems, and uses the floor function in solving them, for a total cost of $O(n)$ expected time.
Later, in 1979, Fortune and Hopcroft~\cite{Fortune79} gave a deterministic algorithm that uses the floor function, and runs in $O(n\log\log n)$ time.

The bounds above hold as long as the dimension $d$ is constant, as they involve factors that are exponential in $d$. Thus, when $d$ is large (e.g., $d=n$), the problem seems to be much less understood.
Shamos and Bentley~\cite{Bentley76} conjectured in 1976 that, for $d=n$, and under the $L_2$ metric, the problem can be solved in $O(n^2\log n)$ time;
so far, their conjectured bound is considerably far from the $O(n^\omega)$ state-of-the-art time bound for this case~\cite{Indyk04},
where $\omega < 2.373$ denotes the exponent for matrix multiplication (see below).
If one settles on approximate solutions, many efficient algorithms were developed in the last two decades, 
mostly based on LSH (locality sensitive hashing) schemes,
and dimensionality reduction via the Johnson-Lindenstrauss transform;
see~\cite{Andoni08,Ailon09} for examples of such algorithms.
These algorithms are often used for finding {\em approximate nearest neighbors}, which itself is of major importance and in massive use in many practical fields of computer science.
Nevertheless, finding an $exact$ solution seems to be a much harder task.

We consider the case where $d$ depends on $n$, i.e., $d=n^r$ for some $ r > 0 $.
Note that a naive brute-force algorithm runs in $O(n^2 d)$ time and works for any metric $L_\tau$.
For some $L_\tau$ metrics, a much faster solution is known; see~\cite{Indyk04}.
Specifically, the $\mathrm{L_2}$ \CP{} problem can be solved by one algebraic matrix multiplication, so for example when $d=n$,
it can be solved in $O(n^\omega)$ time (as already mentioned above).
If $\tau \geq 2$ is an {\em even} integer, then $L_\tau$ \CP{} can be solved in $O(\tau n^\omega)$ time.
However, for other $L_\tau$ metrics, such as when $\tau$ is {\em odd} (or fractional), or  the $L_\infty$ metric, the known solutions are significantly inferior.

For the $L_1$ and $L_\infty$ metrics, Indyk {\em et al.}~\cite{Indyk04} obtained the first (and best known until now) non-naive algorithms for the case $d = n$.
For $L_1$, they gave an algorithm that runs in $O\left(n^{\frac{\omega+3}{2}}\right) = O(n^{2.687})$ time,
and for $L_\infty$, one that runs in $O\left(n^{\frac{\omega+3}{2}}\log D\right) = O(n^{2.687}\log D)$ time,
where $D$ is the diameter of the given point-set. The bound for $L_\infty$ is {\em weakly polynomial}, due to the dependence on $D$,
and, for real data, only yields an approximation. 
Their paper is perhaps the most related to our work.

Our new approach is based on two main observations. The first is showing a reduction from $\linfty$ \CPD{} to another well-studied problem,
{\em dominance product}. The second is by showing we can solve the optimization problem deterministically by executing the decision procedure only $O(\log n)$ times.

We also give improved runtime analysis for the dominance product problem, defined as follows.
\setdescription{leftmargin=1.5cm,labelindent=\parindent}
\begin{description}
	\item[{\bf Dominance Product:}]
	given a set $S$ of $n$ points ${p_1, \ldots, p_n}$ in $\mathbb{R}^d$, compute a matrix $D$ such that for each $i, j \in[n]$,
	$
	D[i,j] = \Bigl| \{k \mid p_i[k] \leq p_j[k]\} \Bigr|. 
	$
\end{description} 
This matrix is called the {\em dominance product} or {\em dominance matrix} for $S$.
For $d = n$, there is a non-trivial strongly subcubic algorithm by \Matousek~\cite{Matousek91} (see Section~\ref{sec:dominance}), and a slightly improved one by Yuster~\cite{Yuster09}.
For $d \leq n$, there are extensions of \Matousek's algorithm by Vassilevska-Williams, Williams, and Yuster~\cite{VassilevskaWY09}.
All of them use fast matrix multiplications.

Dominance product computations were liberally used to improve some fundamental algorithmic problems. For example,
Vassilevska-Williams, Williams, and Yuster~\cite{VassilevskaWY09},
give the first strongly subcubic algorithm for the {\em all pairs bottleneck paths} (APBP) problem, using dominance product computations.
Duan and Pettie~\cite{DuanP09} later improved their algorithm, also by using dominance product computations,
in fact, their time bound for $(\max,\, \min)$-product match the current time bound of computing the dominance product of $n$ points in $\reals^n$.
Yuster~\cite{Yuster09} showed that APSP can be solved in strongly subcubic time if the number of distinct weights of edges emanating from
any fixed vertex is $O(n^{0.338})$. In his algorithm, he uses dominance product computation as a black box.


\subsection{Preliminaries}\label{sec:pre}
We review some notations that we will use throughout the paper.
We denote by $[N] = \{1,\ldots, \ceil{N} \}$, the set of the first $\ceil{N}$ natural numbers succeeding zero, for any $N\in \reals^+$.
For a point $p \in \reals^d$, we denote by $p[k]$ the $k$-th coordinate of $p$, for $k\in [d]$.
For a matrix $A$, we denote the {\em transpose} of $A$ by $A^T$.
The $\tilde{O}(\cdot)$ notation hides poly-logarithmic factors.

Most of the algorithms discussed in this paper heavily rely on fast matrix multiplication algorithms.
Throughout the paper, $\omega < 2.373$ denotes the exponent of multiplying two $n\times n$ matrices~\cite{Williams12, Gall14}, and
$\omega(1,r,1)$ refers to the exponent of multiplying an $n\times n^r$ matrix by an $n^r \times n$ matrix, for some $r > 0$; see~\cite{Huang98, Gall12}. 
For more details on rectangular matrix multiplication exponents,
we refer the reader to the seminal work of Huang and Pan~\cite{Huang98}, and to a more recent work of Le Gall~\cite{Gall12,Gall12Arxiv}.

\subsection{Our Results}\label{sec:results}
Let $DP(n,d)$ denote the runtime order for computing the dominance product (defined above) of $n$ points in $\reals^d$.
We obtain the following results for the $\linfty$ \CP{} problem in $\reals^d$, where $d=n^r$, for some $r > 0$.
\begin{theorem}\label{thm:deterministic}
	$\linfty$ {\textnormal \CP{}} can be solved by a deterministic algorithm that runs in $O(DP(n,d)\log n)$ time.
\end{theorem}
Theorem~\ref{thm:deterministic} improves the $O(n^{2.687}\log D)$ time bound of Indyk {\em et al.}~\cite{Indyk04} (see above) in two aspects,
first is that the polynomial factor $n^{2.687}$ goes slightly down to $DP(n,n)=n^{2.684}$, which we then improve further to $n^{2.6598}$
in Theorem~\ref{thm:dominance}; this holds also for Theorem~\ref{thm:randomized}, stated below.
The second aspect is that the $\log D$ factor is replaced by a $\log n$ factor, which makes our algorithm strongly-polynomial,
independent of the diameter of the given point-set.

For the proof of Theorem~\ref{thm:deterministic}, we first show a reduction from $\linfty$ \CPD{} to dominance product computation, then
we show that the optimization problem can be cleverly solved deterministically by executing the decision procedure only $O(\log n)$ times.
\begin{theorem}\label{thm:randomized}
	$\linfty$ {\textnormal \CP{}} can be solved by a randomized algorithm that runs in $O(DP(n,d))$ expected time.
\end{theorem}
\begin{theorem}\label{thm:integers}
	For points with integer coordinates from $[-M,M]$,	$\linfty$ {\textnormal \CP{}} can be solved by a deterministic algorithm that runs in $\tilde{O}\left(\min\{Mn^{\omega(1,r,1)},\, DP(n,d) \} \right)$ time.
\end{theorem}
From Theorem~\ref{thm:integers} we obtain that for $n$ points in $\reals^n$ with small integer coordinates we can solve the
{\em optimization} problem in $O(n^\omega)$ time,
which is a significant improvement compared to the general case from Theorems~\ref{thm:deterministic} and~\ref{thm:randomized}.

Additionally, in Theorem~\ref{thm:dominance} we give improved bounds for $DP(n,d)$. 
\begin{theorem}\label{thm:dominance}
given a set $S$ of $n$ points ${p_1, \ldots, p_n}$ in $\mathbb{R}^d$,
their dominance product can be computed in $O(DP(n,d))$ time, where 
\begin{align*}\label{DP}
&DP(n,d) \leq
\begin{cases}
d^{0.697}n^{1.896} + n^{2+o(1)}  &\text{if }\, d\leq n^{\frac{\omega - 1}{2}} \leq n^{0.687} \\
d^{0.909}n^{1.75} &\text{if } n^{0.687} \leq d \leq n^{0.87} \\
d^{0.921}n^{1.739} &\text{if } n^{0.87} \leq d \leq n^{0.963} \\
d^{0.931}n^{1.73} &\text{if } n^{0.963} \leq d \leq n^{1.056}
\end{cases}
\end{align*} 	
\end{theorem}
In particular, we obtain that $DP(n,n) = n^{2.6598}$, which improves Yuster's $O(n^{2.684})$ time bound. 
As mentioned above, these bounds will slightly improve the time bounds for algorithms that use dominance product computation as a bottleneck step
(see some examples above). In the rest of the paper we will often refer to the function $DP(n,d)$ above.

\section{$\linfty$ Closest Pair}\label{sec:closest_reals}
Recall that, given a set $S$ of $n$ points $p_1,\ldots, p_n$ in $\reals^d$, the $\linfty$ \CP{} problem is to find a pair of points $(p_i, p_j)$,
such that $i\neq j$ and $\left\|p_i - p_j \right\|_\infty = \min_{\ell \neq m \in[n]}{\left\|p_\ell - p_m \right\|_\infty}$.
The corresponding decision version of this problem is to determine whether there is a pair of distinct points $(p_i, p_j)$ such that $\left\|p_i - p_j \right\|_\infty \leq \delta$,
for a given $\delta >0$.

Naively, we can compute all the distances between every pair of points in $O(n^2d)$ time, and choose the smallest one.
However, as we see next, a significant improvement can be achieved, for any $d=n^r$, for any $r > 0$.

Specifically, we first obtain the following theorem.
\begin{theorem}\label{thm:decision}
	Given a parameter $\delta >0$, and a set $S$ of $n$ points $p_1,\ldots, p_n$ in $\reals^d$,
	the set of all pairs $(p_i, p_j)$ with $\left\|p_i - p_j \right\|_\infty \leq \delta$, 
	can be computed in $O(DP(n,d))$ time.
\end{theorem}

\begin{proof}
	First, we note the following trivial but useful observation.
	
	\begin{observation}\label{col:fredman}
		For a pair of points $p_i, p_j \in \reals^d$,
		$\left\|p_i - p_j \right\|_\infty \leq \delta$  $\Longleftrightarrow$
		$p_i[k] \leq p_j[k] + \delta$ and $p_j[k] \leq p_i[k] + \delta$, for every coordinate $k\in[d]$.
	\end{observation} 
	
	Indeed, a pair of points $(p_i, p_j)$ satisfies $\left\|p_i - p_j \right\|_\infty = \max_{k\in [d]} {\left| p_i[k] - p_j[k]\right|}  \leq \delta$ 
	$\Longleftrightarrow$ for every coordinate $k\in[d]$,
	$\left|p_i[k] - p_j[k]\right|\leq \delta$. The last inequalities hold iff
	$p_i[k] - p_j[k] \leq \delta$ and $p_j[k] - p_i[k] \leq \delta$, or, equivalently, iff
	$p_i[k] \leq p_j[k] + \delta$ and $p_j[k] \leq p_i[k] + \delta$, for each $k\in [d]$.
	Although the rephrasing in the observation is trivial, it is crucial for our next step.
	It can be regarded as a (simple) variant of what is usually referred to as ``Fredman's trick" (see~\cite{Fredman76}).

	For every $i\in [n]$ we create a new point $p_{n+i} = p_i + (\delta, \delta, \ldots, \delta)$. 
	Thus in total, we now have $2n$ points. Concretely, for every $i\in [n]$, we have the points
	\[
	\begin{array}{ccccccc}
	p_i &= \bigl( &p_i[1], &p_i[2],&\ldots, &p_i[d] &\bigr),  \\
	p_{n+i} &= \bigl( &p_{i}[1]+\delta, &p_{i}[2]+\delta, &\ldots, &p_{i}[d] + \delta &\bigr).
	\end{array}
	\]
	We compute the dominance matrix $D_\delta$ for these $2n$ points, using the algorithm from Section~\ref{sec:improved_dominance}.
	By Observation~\ref{col:fredman}, a pair of points $(p_i, p_j)$ satisfies
	\[ \left\|p_i - p_j \right\|_\infty \leq \delta
	\Longleftrightarrow \left( D_\delta[i,n+j] = d \right) \: \wedge \: \left( D_\delta[j,n+i] =d \right),
	\]
	so we can find all these pairs in $O(n^2)$ additional time.
	Clearly, the runtime is determined by the time bound of computing the dominance matrix $D_\delta$, that is, $O(DP(n,d))$.
\end{proof}

The proof of Theorem~\ref{thm:decision} shows that solving the $\linfty$ \CPD{} is not harder than computing the dominance matrix for $n$ points in $\reals^d$.
In particular, by the decision tree complexity bound for computing dominance matrices, as discussed in Section~\ref{sec:dominance}, the following result is straightforward.
\begin{corollary}\label{cor:decision_tree}
	Given a parameter $\delta > 0$, and a set $S$ of $n$ points $p_1,\ldots, p_n$ in $\reals^d$, determining all pairs $i\neq j$ such that $\left\|p_i - p_j \right\|_\infty \leq \delta$ 
	can be done using $O(d n\log n)$ pairwise comparisons (of real numbers).
\end{corollary}
By  Corollary~\ref{cor:decision_tree}, we obtain that the 2-linear decision tree complexity for the $\linfty$ \CPD{} problem is $O(d n\log n)$.
This bound matches a special case of an old conjectured algorithmic complexity bound by Shamos and Bentley (see Section~\ref{sec:intro}, and~\cite{Bentley76}).

\subsection{Solving the Optimization Problem}
The algorithm from Theorem~\ref{thm:decision} solves the $\linfty$ \CPD{} problem.
It actually gives a stronger result, as it finds {\em all} pairs of points $(p_i, p_j)$ such that $\left\|p_i - p_j \right\|_\infty \leq \delta$.
We use this algorithm in order to solve the optimization problem $\linfty$ \CP{}.

As a ``quick and dirty" solution, one can solve the optimization problem by using the algorithm from Theorem~\ref{thm:decision} 
to guide a binary search over the diameter $W$ of the input point set, which is at most twice the largest absolute value of the 
coordinates of the input points.
If the coordinates are integers then we need to invoke the algorithm from Theorem~\ref{thm:decision} $O(\log W)$ times.
If the coordinates are reals, we invoke it $O(B)$ times for $B$ bits of precision.
However, the dependence on $W$ makes this method weakly polynomial, and, for real data, only yields an approximation. 
As we show next, this naive approach can be replaced by strongly-polynomial algorithms, A deterministic one that runs in $O(DP(n,d)\log n)$ time, 
and a randomized one that runs in $O(DP(n,d))$ expected time.

\subparagraph*{Deterministic strongly-polynomial algorithm.}

\begin{theorem}
Given a set $S$ of $n$ points $p_1,\ldots, p_n$ in $\reals^d$, the $\linfty$ {\rm \CP{}} problem can be solved for $S$ in $O(DP(n,d) \log n)$ time.
\end{theorem}

\begin{proof}
Since the distance between the closest pair of points, say $p_i, p_j$, is
\[
\delta_0 = \left\|p_i - p_j \right\|_\infty = \max_{k\in [d]} {\bigl| p_i[k] - p_j[k]\bigr|},
\]
it is one of the $O(n^2d)$ values $p_\ell[k] - p_m[k]$,  $\ell, m \in [n]$, $k\in [d]$.
Our goal is to somehow search through these values, using the decision procedure (i.e., the algorithm from Theorem~\ref{thm:decision}).
However, enumerating all these values takes $\Omega(n^2d)$ time, which is too expensive, and pointless anyway,
since by having them, the closest pair can be found immediately.
Instead, we proceed in the following more efficient manner.

For each $k\in [d]$, we sort the points of $S$ in increasing order of their $k$-th coordinate. 
This takes $O(nd\log n)$ time in total. 
Let $\left(p^{(k)}_1,\ldots, p^{(k)}_n \right)$ denote the sequence of the points of $S$ sorted in increasing order of their $k$-th coordinate.
For each $k$, let $M^{(k)}$ be an $n\times n$ matrix, so that for $i,j\in [n]$, we have 
\[
M^{(k)}[i,j] = p^{(k)}_i[k] - p^{(k)}_j[k].
\]
We are in fact interested only in the upper triangular portion of $M^{(k)}$, where its elements are positive, but for simplicity of presentation, we ignore this issue.
(We view the row indices from bottom to top, i.e., the first row is the bottommost one, and the column indices from left to right.)

Observe that each row of $M^{(k)}$ is sorted in decreasing order and each column is sorted in increasing order.
Under these conditions, the selection algorithm of Frederickson and Johnson~\cite{FREDERICKSON} 
can find the $t$-largest element of $M^{(k)}$, for any $1\leq t \leq n^2$, in $O(n)$ time.\footnote{Simpler algorithms can select the $t$-largest element in such cases in $O(n\log n)$ time, which is also sufficient for our approach.}
(Note that we do not need to explicitly construct the matrices $M^{(k)}$, this will be too expensive.
The bound of Frederickson-Johnson's algorithm holds as long as each entry of $M^{(k)}$ is accessible in $O(1)$ time, like in our case.)

We use this method to conduct a simultaneous binary search over all $d$ matrices $M^{(k)}$ to find $\delta_0$.
At each step of the search we maintain two counters $L_k \leq H_k$, for each $k$.
Initially $L_k = 1$ and $H_k = n^2$.
The invariant that we maintain is that, at each step, $\delta_0$ lies in between the $L_k$-th and the $H_k$-th largest elements of $M^{(k)}$, for each $k$.

Each binary search step is performed as follows.
We compute $r_k = \floor{(L_k + H_k)/2}$, for each $k$, and apply the  Frederickson-Johnson algorithm to retrieve the $r_k$-th largest element of $M^{(k)}$,
which we denote as $\delta_k$, in total time $O(nd)$.
We give $\delta_k$ the weight $H_k - L_k +1$, and compute the weighted median $\delta_{\text{med}}$ of $\{\delta_1, \ldots, \delta_d\}$.
We run the $\linfty$ \CPD{} procedure of Theorem~\ref{thm:decision} on $\delta_{\text{med}}$.
Suppose that it determines that $\delta_0 \leq \delta_{\text{med}}$. Then for each $k$ for which $\delta_k \geq \delta_{\text{med}}$ we know that
$\delta_0 \leq \delta_k$, so we set $H_k := r_k$ and leave $L_k$ unchanged.
Symmetric actions are taken if $\delta_0 > \delta_{\text{med}}$.
In either case, we remove roughly one quarter of the candidate differences;
that is, the sum $\sum_{k\in [d]}{\left(H_k - L_k + 1 \right)}$ decreases by roughly a factor of $3/4$.
Hence, after $O(\log n)$ steps, the sum becomes $O(d)$, and a straightforward binary search through the remaining values finds $\delta_0$.
The overall running time is
\[
O(nd\log n + DP(n,d)(\log n + \log d)).
\]
Since in our setting $d$ is polynomial in $n$, and $nd \ll DP(n,d)$, we obtain that the overall runtime is 	$O(DP(n,d)\log n)$.
This completes the proof of Theorem~\ref{thm:deterministic}.
\end{proof}

\subparagraph*{Randomized algorithm.}
Using randomization, we can improve the time bound of the preceding deterministic algorithm to equal the time bound of computing the dominance product $O(DP(n,d))$ in expectation.
This can be done by using a randomized optimization technique by Chan~\cite{Chan1999}. Among the problems for which this technique can be applied, Chan specifically addresses the \CP{} problem. 
\begin{theorem}[Chan~\cite{Chan1999}]\label{thm:Chan}
	Let $U$ be a collection of objects.	If the {\textnormal \CPD{}} problem can be solved in $O(T(n))$ time, for an arbitrary distance function $d: U\times U \rightarrow \reals$,
	then the {\textnormal \CP{}} problem can be solved in $O(T(n))$ expected time, assuming that $T(n)/n$ is monotone increasing.
\end{theorem}
We refer the reader to~\cite{Chan1999}, for the proof of Theorem~\ref{thm:Chan}.
By Theorem~\ref{thm:decision}, $\linfty$ \CPD{} can be solved in $O(DP(n,d))$ time.
Clearly, $DP(n,d)/n$ is monotone increasing in $n$. Hence, by Theorem~\ref{thm:Chan},
we obtain a randomized algorithm for $\linfty$ \CP{} that runs in $O(DP(n,d))$ expected time, as stated in Theorem~\ref{thm:randomized}.

\section{$\linfty$ Closest Pair with Integer Coordinates}\label{sec:integers}
A considerable part of the algorithm from the previous section is the reduction to computing a suitable dominance matrix.
The algorithms for computing dominance matrices given in Section~\ref{sec:dominance} do not make any assumptions on the coordinates of the points, and support real numbers.
When the coordinates are bounded integers, we can improve the algorithms. In particular, for $n$ points in $\reals^n$ with small integer coordinates we can solve the
{\em optimization} problem in $O(n^\omega)$ time,
which is a significant improvement compared to the $O(n^{2.6598})$ time bound of our previous algorithm for this case\footnote{For integer coordinates that are bounded by a constant, the $L_\infty$-diameter of the points is also a constant (bounded by twice the largest coordinate), hence,
one can use the decision procedure to (naively) guide a binary search over the diameter in constant time.}.
Our improvement is based on techniques for computing $(\min, +)$-matrix multiplication over integer-valued matrices.

\begin{theorem}\label{thm:closest_integers}
Let $S$ be a set of $n$ points $p_1,\ldots, p_n$ in $\reals^d$ such that $d=n^r$ for some $r > 0$,
and for all $i\in[n]$, $k\in [d]$, $p_i[k]$ is an integer in $[-M,M]$.
Then the $L_\infty$ closest pair can be computed in
\[
\tilde{O}\left(\min\left\{Mn^{\omega(1,r,1)},\, DP(n,d) \right\}\right) \text{  time}.
\]
\end{theorem}

We first define  $(\max, +)$-product and $(\min, +)$-product over matrices.


\begin{definition}[Distance products of matrices]
Let $A$ be an $n \times m$ matrix and $B$ be an $m \times n$ matrix.
The $(\max,+)$-product of $A$ and $B$, denoted by $A\star B$, is the $n \times n$ matrix $C$ whose elements are given by
\[
c_{ij} = \max_{1\leq k \leq m}{\{a_{ik} + b_{kj}\}},\quad \text{for } i,j\in [n].
\]
Similarly, the $(\min,+)$-product of $A$ and $B$ denoted by $A\ast B$ is the $n \times n$ matrix $C'$ whose elements are given by
\[
c'_{ij} = \min_{1\leq k \leq m}{\{a_{ik} + b_{kj}\}},\quad \text{for } i,j\in [n].
\]
We refer to either of the $(\min,+)$-product or the $(\max,+)$-product as a {\em distance product}.
\end{definition}

The distance product of an $n \times m$ matrix by an $m \times n$ matrix 
can be computed naively in $O(n^2m)$ time.
When $m=n$, the problem is equivalent to APSP (all pairs shortest paths) problem
in a directed graph with real edge weights, and the fastest algorithm known is a recent one by Chan and Williams~\cite{ChanW16}
that runs in $O\left(n^3 / 2^{\sqrt{\Omega(\log n)}}\right)$ time.
It is a prominent long-standing open problem whether a truly subcubic algorithm for this problem exists.
However, when the entries of the matrices are integers, we can convert distance products of matrices into standard algebraic products.
We use a technique by Zwick~\cite{Zwick02}. 

\begin{lemma}[Zwick~\cite{Zwick02}]\label{lem:Alon}
Given an $n \times m$ matrix $A=\{a_{ij}\}$ and an $m \times n$ matrix $B=\{b_{ij} \}$ such that $m = n^r$ for some $ r > 0$,
and all the elements of both matrices are integers from $[-M,M]$, their $(\min,+)$-product $C = A \ast B$ can be computed in $\tilde{O}(Mn^{\omega(1,r,1)})$ time.
\end{lemma}
%
With minor appropriate modifications, the $(\max, +)$-product of matrices $A$ and $B$ can be computed within the same time as in Lemma~\ref{lem:Alon}.

We now give an algorithm for computing all-pairs $L_\infty$ distances, by using the fast algorithm for computing $(\max,+)$-product over bounded integers.

\begin{lemma}
Let $S$ be a set of $n$ points $p_1,\ldots, p_n$ in $\reals^d$ such that $d=n^r$ for some $r > 0$,
and for all $i\in[n]$, $p_i[k]$ is an integer from the interval $[-M,M]$, for all $k\in[d]$.
Then the $L_\infty$ distances between all pairs of points $(p_i, p_j)$ from $S$ can be computed in $\tilde{O}(Mn^{\omega(1,r,1)})$ time.
\end{lemma}
\begin{proof}
We create the $n \times d$ matrix $A = \{a_{ik}\}$ and the $d \times n$ matrix $B = (-A)^T = \{b_{ki}\}$, where
\begin{align*}
a_{ik} = p_i[k], \quad &\text{for } i\in [n],\, k\in [d] \\[5pt]
\medskip
b_{ki} = - p_i[k], \quad &\text{for } i\in[n],\, k\in[d].
\end{align*}

Now we compute the $(\max,+)$-product $C = A \star B$.
The matrix $L$ of all-pairs $L_\infty$-distances is then easily seen to be
\[
L[i,j]= \max\bigl\{C[i,j], C[j,i]\bigr\} = \left\|p_i - p_j \right\|_\infty,
\]
for every pair $i,j\in [n]$.

Clearly, the runtime is determined by computing the $(\max,+)$-product  $C = A\star B$. This is done as explained earlier, and achieves the required running time. 
\end{proof}

Consequently, by taking the minimum from the algorithm above, and the (say, deterministic) algorithm from Section~\ref{sec:closest_reals}, 
we obtain that for points in $\reals^d$ with integer coordinates from $[-M, M]$, where $d=n^r$ for some $r > 0$, we can find the $L_\infty$ closest pair in
\[
\tilde{O}\left(\min\left\{Mn^{\omega(1,r,1)},\, DP(n,d) \right\}\right) \text{  time},
\]
as stated in Theorem~\ref{thm:integers}.

\section{Dominance Products}\label{sec:dominance}

We recall the dominance product problem: given $n$ points $p_1, \ldots, p_n$ in $\reals^d$, we want to compute a matrix $D$ such that for each $i, j \in[n]$,
\[
D[i,j] = \Bigl| \{k \mid p_i[k] \leq p_j[k]\} \Bigr|. 
\] 
It is easy to see that the matrix $D$ can be computed naively in $O(dn^2)$ time.
Note that, in terms of decision tree complexity, it is straightforward to show that $O(d n\log n)$ pairwise comparisons suffice for computing the dominance product of $n$ points in $\mathbb{R}^d$.
However, the actual best known time bound to solve this problem is significantly larger than its decision tree complexity bound.

The first who gave a truly subcubic algorithm to compute the dominance product of $n$ points in $\reals^n$ is \Matousek{}~\cite{Matousek91}. 
We first outline his algorithm, and then present our extension and improved runtime analysis.

\begin{theorem}[\Matousek{}~\cite{Matousek91}]\label{thm:Matousek}
	Given a set $S$ of $n$ points in $\mathbb{R}^n$,
	the dominance matrix for $S$ can be computed in $O(n^{\frac{3+\omega}{2}}) = O(n^{2.687})$ time.
\end{theorem}
\begin{proof}
	For each $j \in [n]$, sort the $n$ points by their $j$-th coordinate.
	This takes a total of $O(n^2 \log n)$ time.
	Define the {\em $j$-th rank of point} $p_i$, denoted as $r_j(p_i)$, to be the position of $p_i$ in the sorted list for coordinate $j$.
	Let $s \in [\log n, n]$ be a parameter to be determined later.
	Define $n/s$ pairs (assuming for simplicity that $n/s$ is an integer) of $n \times n$ Boolean matrices $(A_1, B_1),\ldots, (A_{n/s}, B_{n/s})$ as follows:
	\begin{align*}
	&A_k[i,j] = 
	\begin{cases}
	1  & \text{if }\, r_j(p_i) \in [ks, ks+s) \\
	0  & \text{otherwise}, 
	\end{cases}
	\medskip
	&B_k[i,j] = 
	\begin{cases}
	1 & \text{if } r_j (p_i) \geq ks+s \\
	0 & \text{otherwise},
	\end{cases}
	\end{align*}
	for $i,j \in [n]$.
	Put $C_k = A_k \cdot B^T_k$.
	Then $C_k[i,j]$ equals the number of coordinates $t$ such that $r_t(p_i) \in [ks, ks+s)$, and $r_t(p_j) \geq ks+s$.
	
	Thus, by letting 
	$
	C = \sum^{n/s}_{k=1} C_k
	$,
	we have that $C[i, j]$ is the number of coordinates $t$ such that $p_i[t] \leq p_j[t]$ and $\floor{r_t(p_i)/s} < \floor{r_t(p_j)/s}$.
	
	Next, we compute a matrix $E$ such that $E[i,j]$ is the number of coordinates $t$ such that $p_i[t] \leq ~ p_j[t]$ and $\floor{r_t(p_i)/s} = \floor{r_t(p_j)/s}$.
	Then $D := C + E$ is the desired dominance matrix.
	
	To compute $E$, we use the $n$ sorted lists we computed earlier.
	For each pair $(i, j) \in [n] \times [n]$, we retrieve $q:= r_j(p_i)$.
	By reading off the adjacent points that precede $p_i$ in the $j$-th sorted list in reverse order
	(i.e., the points at positions $q - 1$, $q - 2$, etc.), and stopping as soon as we reach a point $p_k$
	such that $\floor{r_j(p_k)/s} < \floor{r_j(p_i)/s}$,
	we obtain the list $p_{i_1}, \ldots, p_{i_l}$ of $l \leq s$ points such that
	$p_{i_x}[j] \leq p_i[j]$ and $\floor{r_j (p_i)/s} = \floor{r_j (p_{i_x})/s}$.
	For each $x = 1, \ldots, l$, we add a $1$ to $E[i_x, i]$.
	Assuming constant time lookups and constant time probes into a matrix (as is standard in the real RAM model),
	this entire process takes only $O(n^2s)$ time.
	The runtime of the above procedure is therefore $O(n^2s +\frac{n}{s}\cdot n^\omega)$.
	Choosing $s = n^{\frac{\omega-1}{2}}$, the time bound becomes $O(n^{\frac{3+\omega}{2}})$. 
\end{proof}

Yuster~\cite{Yuster09} has slightly improved this algorithm to run in $O(n^{2.684})$ time, by using rectangular matrix multiplication. 

\subsection{Generalized and Improved Bounds}\label{sec:improved_dominance}
We extend Yuster's idea to obtain bounds for dimension $d=n^r$, for the entire range $r > 0$,
and, at the same time, give an improved time analysis, using the recent bounds for rectangular matrix multiplications of Le Gall~\cite{Gall12, Gall12Arxiv} coupled with an interpolation technique. This analysis is not trivial, as Le Gall's bounds for $\omega(1,r,1)$ are obtained by a nonlinear optimization problem,
and are only provided for a few selected values of $r$ (see Table 1 in~\cite{Gall12Arxiv} and~\cite{Gall12}).
Combining Le Gall's exponents with an interpolation technique, similar to the one used by Huang and Pan~\cite{Huang98}, we obtain improved bounds for all values $d = n^r$,
for any $r > 0$.

Note that the matrices $A_k$ and $B_k$, defined above, are now  $n \times d$ matrices.
Thus, the sum $C$ defined earlier, can be viewed as a product of block matrices
\[
C =
\begin{bmatrix}
A_1 &A_2  &\cdots  &A_{n/s}
\end{bmatrix}
\cdot
\begin{bmatrix} 
B^T_1 \\[0.5ex]
B^T_2 \\
\vdots \\
B^T_{n/s}
\end{bmatrix}.
\]
Thus, to compute $C$ we need to multiply an $n \times \left( dn/s \right)$ matrix by a $\left(dn/s \right) \times n$ matrix.
Computing $E$ in this case can be done exactly as in \Matousek{}'s algorithm, in $O(nds)$ time.

Consider first the case where $d$ is small; concretely, $d \leq n^{\frac{\omega -1}{2}}$.
In this case we compute $C$ using the following result by Huang and Pan.
\begin{lemma}[Huang and Pan~\cite{Huang98}]\label{lem:hp}
	Let $\alpha = \sup\bigl\{0 \leq r \leq 1 \mid w(1,r,1) = 2 + o(1)  \bigr\}$. Then for all $n^\alpha \leq m\leq n$, 
	one can multiply an $n \times m$ matrix with an $m \times n$ matrix in time
	$
	O\left( m^{\frac{\omega -2}{1-\alpha}} n^{\frac{2-\omega\alpha}{1-\alpha}}  \right).
	$	
\end{lemma}
Huang and Pan~\cite{Huang98} showed that $\alpha > 0.294$.
Recently, Le Gall~\cite{Gall12, Gall12Arxiv} improved the bound on $\alpha$ to
$\alpha > 0.302$. By plugging this into Lemma~\ref{lem:hp},
we obtain that multiplying an $n \times m$ matrix with an $m \times n$ matrix, where  $n^\alpha \leq m\leq n$, can be done in time 
$
O(m^{0.535} n^{1.839}).
$ 

From the above, computing $C$ and $E$ can be done in
$O\left( (dn/s)^{0.535} n^{1.839} +dns \right)$ time.
By choosing $s = {n^{0.896}} / {d^{0.303}}$, the runtime is asymptotically minimized, and we obtain the time bound
$O(d^{0.697}n^{1.896})$. This time bound holds only when 
$n^\alpha < n^{0.302} \leq dn/s \leq n$, which yields the time bound
\[
O(d^{0.697}n^{1.896} +n^{2+o(1)}), \text{ for } d \leq n^{(\omega - 1) /2} \leq n^{0.687}.
\]

We now handle the case $d > n^{(\omega - 1) /2}$. Note that in this case,  $dn/s > n$ (for $s$ as above), thus, we cannot use the bound from Lemma~\ref{lem:hp}.
Le Gall~\cite{Gall12, Gall12Arxiv} gives a table (Table 1 in~\cite{Gall12Arxiv} and~\cite{Gall12}) of values $r$ (he refers to them as $k$),
including values of $r > 1$ (which is what we need), with various respective exponents $\omega(1,r,1)$.
We will confine ourselves to the given bounds for the values $r_1=1.1$, $r_2=1.2$, $r_3=1.3$, and $r_4=1.4$.
We denote their corresponding exponents $\omega(1,r_i,1)$ by $\omega_1 \leq 2.456151, \omega_2 \leq 2.539392$, $\omega_3 \leq 2.624703$, and $\omega_4 \leq 2.711707$ respectively.
The exponent for $r_0 = 1$ is $\omega_0 = \omega \leq 2.372864$ (see~\cite{Williams12, Gall14}).
\def\om{{\omega}}
\def\ze{{\zeta}}

\begin{table}
	\begin{center}
		\begin{tabular}{|c|c|c|}
			\hline
			$r$ & $\om$ & $\ze$ \\
			\hline
			$r_0 = 1.0$ & $\om_0 = 2.372864$ & $\ze_0 = 0.6865$ \\
			$r_1 = 1.1$ & $\om_1 = 2.456151$ & $\ze_1 = 0.7781$ \\
			$r_2 = 1.2$ & $\om_2 = 2.539392$ & $\ze_2 = 0.8697$ \\
			$r_3 = 1.3$ & $\om_3 = 2.624703$ & $\ze_3 = 0.9624$ \\
			$r_4 = 1.4$ & $\om_4 = 2.711707$ & $\ze_4 = 1.0559$\\
			\hline
		\end{tabular}
	\end{center}
	\caption{The relevant entries from Le Gall's table. The dominance product can be computed in $O(n^{\om_i})$ time, for dimension $d_i=n^{\ze_i}$.} 
	\label{table1}
\end{table}

The algorithm consists of two parts. For a parameter $s$, that we will fix shortly,
the cost of computing $C= A\cdot B^T$ is $O\left(n^{\om_r}\right)$,
where $\om_r$ is a shorthand notation for $\om(1,r,1)$, and where $n^r=dn/s$, 
and the cost of computing $E$ is $O(nds) = O\left(s^2n^r\right)$. Dropping the constants 
of proportionality, and equating the two expressions, we choose
$$
s = n^{(\om_r-r)/2} ,\quad\quad\text{that is,}\quad\quad
d = sn^{r-1} = n^{(\om_r+r)/2-1} = n^{\ze_r} ,
$$
for $\ze_r = (\om_r+r)/2-1$. Put $\ze_i = \ze_{r_i}$, for the values $r_0,\ldots,r_4$ mentioned earlier;
see Table~\ref{table1}.

Now if we are lucky and $d=n^{\ze_i}$, for $i=0,1,2,3,4$, then the overall cost of the algorithm is $O(n^{\om_i})$.
For in-between values of $d$, we need to interpolate, using the following bound, which is
derived in the earlier studies (see, e.g., Huang and Pan~\cite{Huang98}), and which asserts that,
for $a\le r\le b$, we have
\begin{equation} \label{interpol}
\om_r \le \frac{ (b-r)\om_a + (r-a)\om_b } {b-a} .
\end{equation}
That is, given $d = n^\ze$, where $\ze_i \le \ze \le \ze_{i+1}$, for some $i\in\{0,1,2,3\}$,
the cost of the algorithm will be $O\left(n^{\om_r}\right)$, where $r$ satisfies
$$
\ze = \ze_r = \frac{\om_r+r}{2}-1 .
$$
Substituting the bound for $\om_r$ from (\ref{interpol}), with $a=r_i$ and $b=r_{i+1}$, we have
$$
\frac{ (r_{i+1}-r)\om_i + (r-r_i)\om_{i+1} } {r_{i+1}-r_i} + r = 2(\ze+1) .
$$
Eliminating $r$, we get
\begin{equation} \label{eq:r}
r = \frac{ 2(\ze+1)(r_{i+1}-r_i) - r_{i+1}w_i + r_iw_{i+1} }
{w_{i+1}+r_{i+1}-w_i-r_i} ,
\end{equation}
and the cost of the algorithm will be $O\left(n^{\om_r}\right)$, where 
\begin{equation} \label{eq:omr}
\om_r \leq \frac{ (r_{i+1}-r)\om_i + (r-r_i)\om_{i+1} } {r_{i+1}-r_i} .
\end{equation}
Note that $r$ is a linear function of $\ze$, and so is $\om_r$. Writing $\om_r = u\ze+v$,
the cost is
\[
O\left(n^{\om_r}\right) = O\left(n^{u\ze+v}\right) = O\left(d^u n^v\right) .
\]

\begin{table}[t]
	\begin{center}
		\begin{tabular}{|c|c|c|c|}
			\hline
			$\ze_{\rm min}$ & $\ze_{\rm max}$ & $u$ & $v$ \\
			\hline
			$0.687$ & $0.87$ & $0.909$ & $1.75$ \\
			$0.87$ & $0.963$ & $0.921$ & $1.739$ \\
			$0.963$ & $1.056$     & $0.931$ & $1.73$\\
			\hline
		\end{tabular}
	\end{center}
	\caption{The time bound for computing dominance product for $n$ points in dimension $n^{\ze_{\rm min}} \leq d \leq n^{\ze_{\rm max}}$ is 
		$O\left(d^u n^v\right)$.}
	\label{table2}
\end{table}

The values of $u$ and $v$ for each of our intervals are given in Table~\ref{table2}.
(The first row covers the two intervals $1.0 \leq r \leq 1.1$ and $1.1 \leq r \leq 1.2$, as the bounds happen to coincide there.)
See also~(\ref{DP}) in Section~\ref{sec:results}.
We have provided explicit expressions for $DP(n,d)$ only for $d \leq n^{\ze_4} = n^{1.056}$, 
which includes the range $d \leq n$, which is the range one expects in practice. 
Nevertheless, the recipe that we provide can also be applied to larger values of $d$, using larger entries from Le Gall's table~\cite{Gall12,Gall12Arxiv}.
Dropping constant factors, we denote the time bound for computing the dominance product of $n$ points in $\reals^d$ by $DP(n,d)$;
see Theorem~\ref{thm:dominance} in Section~\ref{sec:results}.
by plugging the corresponding values of $0.302 < r < 1$ from Le Gall's Table 1 in~\cite{Gall12Arxiv}.
We also note that, for $d=n$, the time bound is $O(n^{2.6598})$, which improves Yuster's $O(n^{2.684})$ time bound mentioned above.



\bibliography{ClosestPair}

\begin{thebibliography}{10}

\bibitem{Ailon09}
Nir Ailon and Bernard Chazelle.
\newblock The fast {J}ohnson-{L}indenstrauss transform and approximate nearest
  neighbors.
\newblock {\em SIAM J. Comput.}, 39(1):302--322, 2009.

\bibitem{Andoni08}
Alexandr Andoni and Piotr Indyk.
\newblock Near-optimal hashing algorithms for approximate nearest neighbor in
  high dimensions.
\newblock {\em Commun. ACM}, 51(1):117--122, 2008.

\bibitem{Ben-Or83}
Michael Ben-Or.
\newblock Lower bounds for algebraic computation trees.
\newblock In {\em Proc. of the 15th Annu. ACM Sympos. on Theory of Computing
  (STOC)}, pages 80--86, 1983.

\bibitem{Bentley80}
Jon~Louis Bentley.
\newblock Multidimensional divide-and-conquer.
\newblock {\em Commun. ACM}, 23(4):214--229, 1980.

\bibitem{Bentley76}
Jon~Louis Bentley and Michael~Ian Shamos.
\newblock Divide-and-conquer in multidimensional space.
\newblock In {\em Proc. of the 8th Annu. ACM Sympos. on Theory of Computing
  (STOC)}, pages 220--230, 1976.

\bibitem{Chan1999}
T.~M. Chan.
\newblock Geometric applications of a randomized optimization technique.
\newblock {\em Discrete {\&} Computational Geometry}, 22(4):547--567, 1999.

\bibitem{ChanW16}
Timothy~M. Chan and Ryan Williams.
\newblock Deterministic {APSP}, orthogonal vectors, and more: Quickly
  derandomizing {R}azborov-{S}molensky.
\newblock In {\em Proc. of the 27th Annu. ACM-SIAM Sympos. on Discrete
  Algorithms (SODA)}, pages 1246--1255, 2016.

\bibitem{DuanP09}
Ran Duan and Seth Pettie.
\newblock Fast algorithms for (max, min)-matrix multiplication and bottleneck
  shortest paths.
\newblock In {\em Proc. of the 20th Annu. {ACM-SIAM} Sympos. on Discrete
  Algorithms (SODA)}, pages 384--391, 2009.

\bibitem{Fortune79}
Steve Fortune and John Hopcroft.
\newblock A note on {R}abin's nearest-neighbor algorithm.
\newblock {\em Inform. Process. Lett.}, 8(1):20--23, 1979.

\bibitem{FREDERICKSON}
Greg~N. Frederickson and Donald~B. Johnson.
\newblock The complexity of selection and ranking in x + y and matrices with
  sorted columns.
\newblock {\em Journal of Computer and System Sciences}, 24(2):197 -- 208,
  1982.

\bibitem{Fredman76}
M.~L. Fredman.
\newblock How good is the information theory bound in sorting?
\newblock {\em Theoret. Comput. Sci}, 1(4):355--361, 1976.

\bibitem{Huang98}
Xiaohan Huang and Victor~Y. Pan.
\newblock Fast rectangular matrix multiplication and applications.
\newblock {\em J. Complexity}, 14(2):257--299, 1998.

\bibitem{Indyk04}
Piotr Indyk, Moshe Lewenstein, Ohad Lipsky, and Ely Porat.
\newblock Closest pair problems in very high dimensions.
\newblock In {\em Proc. 31st International Colloquium on Automata, Languages
  and Programming (ICALP)}, pages 782--792, 2004.

\bibitem{Gall14}
Fran\c{c}ois {Le Gall}.
\newblock Powers of tensors and fast matrix multiplication.
\newblock In {\em Proc. 39th International Sympos. on Symbolic and Algebraic
  Computation (ISSAC)}, pages 296--303, 2014.

\bibitem{Gall12}
Fran{\c{c}}ois {Le Gall}.
\newblock Faster algorithms for rectangular matrix multiplication.
\newblock In {\em Proc. 53rd Annu. {IEEE} Sympos. on Foundations of Computer
  Science (FOCS)}, pages 514--523, 2012.

\bibitem{Gall12Arxiv}
Fran{\c{c}}ois {Le Gall}.
\newblock Faster algorithms for rectangular matrix multiplication.
\newblock {\em CoRR}, abs/1204.1111, 2012.

\bibitem{Matousek91}
Ji\v{r}\'{\i} Matou\v{s}ek.
\newblock Computing dominances in ${E}^n$.
\newblock {\em Inform. Process. Lett.}, 38(5):277--278, 1991.

\bibitem{Preparata85}
Franco~P. Preparata and Michael~I. Shamos.
\newblock {\em Computational Geometry: An Introduction}.
\newblock Springer-Verlag New York, NY, 1985.

\bibitem{Rabin76}
Michael Rabin.
\newblock Probabilistic algorithms.
\newblock In {\em Algorithms and Complexity, Recent Results and New
  Directions}, Academic Press, pages 21--39, 1976.

\bibitem{Shamos75}
Michael~Ian Shamos.
\newblock Geometric complexity.
\newblock In {\em Proc. of 7th Annu. ACM Sympos. on Theory of Computing
  (STOC)}, pages 224--233, 1975.

\bibitem{VassilevskaWY09}
Virginia Vassilevska, Ryan Williams, and Raphael Yuster.
\newblock All pairs bottleneck paths and max-min matrix products in truly
  subcubic time.
\newblock {\em Theory of Computing}, 5(1):173--189, 2009.

\bibitem{Williams12}
Virginia~Vassilevska Williams.
\newblock Multiplying matrices faster than {C}oppersmith-{W}inograd.
\newblock In {\em Proc. 44th Sympos. on Theory of Computing (STOC)}, pages
  887--898, 2012.

\bibitem{Yuster09}
Raphael Yuster.
\newblock Efficient algorithms on sets of permutations, dominance, and
  real-weighted {APSP}.
\newblock In {\em Proc. 20th Annu. {ACM-SIAM} Sympos. on Discrete Algorithms
  (SODA)}, pages 950--957, 2009.

\bibitem{Zwick02}
Uri Zwick.
\newblock All pairs shortest paths using bridging sets and rectangular matrix
  multiplication.
\newblock {\em J. ACM}, 49(3):289--317, 2002.

\end{thebibliography}


\end{document}